\documentclass[11pt,a4paper]{article}
\pdfoutput=1

\usepackage[utf8]{inputenc}
\usepackage{amsmath}
\usepackage{amsfonts}
\usepackage{amsthm}
\usepackage{graphicx}
\usepackage{multirow}

\newcommand{\reals}{\mathbb{R}}
\newcommand{\grad}{\mathrm{grad}}
\newtheorem{theorem}{Theorem}
\newtheorem{assumption}{Assumption}
\newtheorem{corollary}{Corollary}

\begin{document}

\title{On the scalability and convergence of simultaneous parameter
  identification and synchronization of dynamical systems}
\author{Bruno Nery \qquad Rodrigo Ventura \\
  Institute for Systems and Robotics \\
  Instituto Superior Técnico, Lisbon, Portugal}
\date{}
\maketitle

\begin{abstract}
  The synchronization of dynamical systems is a method that allows two
  systems to have identical state trajectories, appart from an error
  converging to zero. This method consists in an appropriate unidirectional coupling
  from one system (drive) to the other (response).
  This requires that the response system shares the same dynamical
  model with the drive.  For the cases where the drive is unknown,
  Chen proposed in 2002 a method to adapt the response system such
  that synchronization is achieved, provided that (1) the response
  dynamical model is linear with a vector of parameters, and (2) there
  is a parameter vector that makes both system dynamics identical.
  However, this method
  has two limitations: first, it does not scale well for complex
  parametric models (e.g., if the number of parameters is greater than
  the state dimension), and second, the model parameters are not
  guaranteed to converge, namely as the synchronization error
  approaches zero. This paper presents an adaptation law 
  addressing these two
  limitations. Stability and convergence proofs, using Lyapunov's
  second method, support the proposed adaptation law. Finally,
  numerical simulations illustrate the advantages of the proposed method, namely
  showing cases where the Chen's method fail, while the proposed one does not.
\end{abstract}

\section{Introduction}
\label{sec:introduction}

Consider two identical continuous time dynamical systems,
designated \emph{drive} (D) and \emph{response} (R). It is well known
that 
the state evolution of each system, when taken separately,
may differ radically if the initial condition for each system differ, namely in the case
of chaotic dynamical systems~\cite{lorenz1963deterministic,grebogi87}.
However, in the presence of a
unidirectional coupling from the drive to
the response system, synchronization of their state trajectories is known
to occur~\cite{rulkov95,kocarev96,pecora97}. In this paper we limit the discussion to
the simplest coupling scheme, in which the response system
receives the full state vector from the drive. In this situation it is
easy to design a controller that synchronizes both systems, using
feedback linearization (Section~\ref{sec:problem-statement}).


Such synchronization assumes that both drive and response have the
same dynamical model. This paper addresses the problem of achieving
synchronization of a response system, when the dynamical model of the
drive is unknown. In particular, we target the problem of
\emph{simultaneous adaptation and synchronization of a response
  system, given an unknown drive.}  Two assumptions are made: (1)~the
response dynamical model depends linearly on a parameter vector, and
(2)~there is a value for this vector that makes both systems
identical. In~2002, Chen and Lü proposed a method to simultaneously
adapt this parameter vector and to make both systems
synchronized~\cite{chen02}. Lyapunov second method was used to prove
the feasibility of this method, however, due to the construction of
the Lyapunov function employed, convergence of the response parameters
is not guaranteed. This has two consequences that prevent the general
usage of this method.  Firstly, it does not scale in complexity: if
the dimension of the parameter vector is greater than the dimension of
the state vector, convergence is not guaranteed. And secondly, even
with a small number of parameters, Chen's proof does not guarantee
effective convergence of the parameters.

In this paper we address both of these problems, presenting a
convergence proof for the simultaneous synchronization and adaptation
of the response to an arbitrary drive system. Moreover, numerical
simulations comparing the proposed approach with Chen's method
illustrate the benefits of the approach.


Chaotic synchronization was first introduced by Pecora and Carrol in
1990~\cite{pecora90}. Since then, many publications have deepend our
knowledge about this
concept~\cite{afraimovich86,kocarev96,pecora97,agiza01}. A method for
synchronizing the Rössler and the Chen chaotic systems using active
control was proposed by Agiza and Yassen~\cite{agiza01}, however the
approach is specific to these particular systems. Chen and Lü proposed
a method to perform simultaneous identification and synchronization of
chaotic systems~\cite{chen02}, but the results show some limitations,
which are discussed in length and addressed in this paper.

The paper is structured a follows: section~\ref{sec:problem-statement}
states formally the problem, followed by the proposed solution in
section~\ref{sec:proposed-solution}; experimental results are presented
in section~\ref{sec:experimental-results}, and
section~\ref{sec:conclusions} concludes the paper.

\section{Problem statement}
\label{sec:problem-statement}

Consider two dynamical systems, called \emph{drive} and
\emph{response}, with a unidirectional coupling between them.
Throughout
this paper we will assume that both drive and response systems are
identical, apart from a parameter vector, which is unknown. The goal of the adaptation law is to determine this
parameter vector. Consider the drive system modeled by
\begin{equation}
  \label{eq:1}
  \dot{x} = f(x) + F(x)\,\theta,
\end{equation}
where $x(t)\in\reals^n$ is the state vector, and $\theta\in\reals^m$
is a parameter vector. The nonlinear functions that support the model
are $f:\reals^n\rightarrow\reals^n$ and
$F:\reals^n\rightarrow\reals^{(n\times m)}$.
The coupling between the drive and the response systems consists in a
bias term, called \emph{synchronization input,} from the drive to the response.
The response system is
identical to the drive, except for the parameter vector~$\alpha\in\reals^m$ and the synchronization input~U,
\begin{equation}
  \label{eq:2}
  \dot{y} = f(y) + F(y)\,\alpha + U(y,x,\alpha),
\end{equation}
where $y(t)\in\reals^n$ is the response state vector, and
$U:\reals^n\times\reals^n\times\reals^m\rightarrow\reals^n$
is the synchronization control function. This function $U$ realizes the
controller that, given the state input from the drive, synchronizes
the response system.

Define the state error $e=y-x$ and the parameter error
$\Delta=\alpha-\theta$; the simultaneous adaptation and
synchronization problem consists in the design of a controller $U$ and
of a parameter adaptation law for $\alpha$ such that both
$\lim_{t\rightarrow\infty} e(t)=0$ and $\lim_{t\rightarrow\infty}
\Delta(t)=0$.

Chen proposes in~\cite{chen02} a solution to this problem
in the form of an adaptation law for $\alpha$.
\begin{assumption}
  There is a
  controller $U$ and a scalar function $V(e)$ that, for
  $\alpha=\theta$, satisfies both (i)~$c_1||e||^2\leq V(e)\leq
  c_2||e||^2$ and (ii)~$\dot{V}(e)\leq -W(e)$, for $c_1,c_2$ positive
  constants, $W(e)$ a positive definite function,
  and
  $U(x,x,u,\theta)=0$. 
\end{assumption}
\noindent
For example, the controller
\begin{equation}
  \label{eq:14}
  U(y,x,\theta) = -e+f(x)-f(y)+\left[F(x)-F(y)\right]\theta,
\end{equation}
and the function $V(e)=\frac{1}{2}e^Te$ satisfy this assumption. 
\begin{theorem}
  Under Assumption~1, the adaptation law
  \begin{equation}
    \label{eq:12}
    \dot{\alpha} = -F^T(x)\left[\grad\,V(e) \right]^T
  \end{equation}
  stabilizes the system at the equilibrium point $e=0$, $\alpha=\theta$.
\end{theorem}
\begin{proof}
  See~\cite{chen02}.
\end{proof}
\noindent
In the proof of this
theorem, Chen employs the Lyapunov function
\begin{equation}
  \label{eq:13}
  V_1(e,\alpha) = \frac{1}{2}e^Te + \frac{1}{2}\Delta^T\Delta.
\end{equation}
There is an hidden assumption in the proof: it
only holds if
$U(y,x,\alpha)-U(y,x,\theta)=\left[F(x)-F(y)\right]\Delta$
(which is true if controller~\eqref{eq:14} is used).

Still, two problems remain that compromise the applicability of this
result. The first one is that~\eqref{eq:12} 
does not guarantee strict definite positiveness of
$-\dot{V}_1$; in particular,
$\dot{V}_1(0,\alpha)=0$ for all values of $\alpha$. This means that, as the synchronization error $e$
approaches zero, the magnitude of the parameter error $\Delta$
is not guaranteed to decrease.
The second problem concerns the null space of
$F^T(x)$: according to~\eqref{eq:12}, the parameter vector $\alpha$
remains changed, as long as $\grad\,V(e)$ lies in the null
space of $F^T(x)$. Taking for instance $V(e)=\frac{1}{2}e^Te$, while the state error
$e$ lies in this null space, the parameter vector
$\alpha$ remains unchanged, even if $\alpha\neq\theta$.

\section{Proposed solution}
\label{sec:proposed-solution}

Let us first obtain a controller function $U$ that achieves
synchronization, assuming that the true
value of the parameter vector is known, $\alpha=\theta$. In this situation, the error
state $e$ has the following dynamics
\begin{equation}
  \label{eq:3}
  \dot{e}=f(y)-f(x) + \left[ F(y)-F(x) \right]\,\theta + U(y,x,\theta).
\end{equation}
Considering now the positive definite Lyapunov function
\begin{equation}
  \label{eq:4}
  V(e) = \frac{1}{2}\,e^Te,
\end{equation}
its time derivative is $\dot{V}=e^T\dot{e}$. Taking the controller
\begin{equation}
  \label{eq:6}
  U(y,x,\theta) = - K\,e - f(y) + f(x) - \left[ F(y) - F(x)
  \right]\,\theta,
\end{equation}
where $K$ is a $(n\times n)$ positive definite matrix, we have that
$\dot{e}=-K\,e$.  Matrix $K$ is thus related with the synchronization
rate.  Since $-\dot{V}=e^TK\,e$ is a positive definite function, for a
positive definite $K$, system \eqref{eq:3} is globally uniformly
asymptotically stable~\cite{sastry99} at the equilibrium point
$e=0$. Note that this controller satisfies Chen's Assumption referred
in the previous section.

Consider now the positive definite Lyapunov function
\begin{equation}
  \label{eq:7}
  V(e,\Delta) = \frac{1}{2}\,e^Te + \frac{1}{2}\,\Delta^T\Delta.
\end{equation}
This function is zero if and only if both the response is synchronized
with the drive, and its parameters equal the drive ones.
The dynamics of the error $e$, while using the controller
\eqref{eq:6}, is then
\begin{equation}
  \label{eq:8}
  \dot{e} = - K\,e + F(x)\,\Delta.
\end{equation}
By left multiplying this equation by $F^T(x)\,L$, where $L$ is a $(n\times
n)$ positive definite matrix (which is related with the adaptation
rate; see below), and transposing the result, one gets the relation
\begin{equation}
  \label{eq:9}
  \Delta^TF^T(x)\,L^T\,F(x) = \dot{e}^TL^TF(x) + e^TK^TL^TF(x).
\end{equation}
We are now in condition to prove the main result of this paper:
\begin{theorem}
  \label{thm:2}
  Assuming that there is a constant matrix $L$ such that
  $G(x)=F^T(x)\,L^T\,F(x)$ is positive definite for all $x$,
  the adaptation law
  \begin{equation}
    \label{eq:10}
    \dot{\alpha} = - F^T(x)\,\left[ \left(L\,K+I\right)\,e + L\,\dot{e} \right],
  \end{equation}
  together with the controller~\eqref{eq:6}, globally uniformly
  stabilizes both the error system~\eqref{eq:8} at~$e=0$, and the
  parameter error at~$\Delta=0$.
\end{theorem}
\begin{proof}
Considering the Lyapunov function~\eqref{eq:7},
we have $\dot{V}=e^T\dot{e}+\Delta^T\dot\Delta$.
Taking the adaptation law~\eqref{eq:10}
together with~\eqref{eq:9}, while noting that $\dot\Delta=\dot\alpha$, one obtains
\begin{equation}
  \label{eq:11}
  \dot{V} = - e^TK\,e - \Delta^TG(x)\Delta.
\end{equation}
Since $G(x)$ is assumed positive definite,
$-\dot{V}$ is also positive definite, from which we can conclude that
$(e,\Delta)=(0,0)$ is a globally uniformly asymptotically
stable~\cite{sastry99} equilibrium point of~\eqref{eq:8}.
\end{proof}
This theorem implies both synchronization ($y=x$) and correct
identification of the parameters ($\alpha=\theta$). Note that the
practical use of the proposed adaptation law~\eqref{eq:10} requires
knowledge of the error time derivative $\dot{e}$, which in principle
can be obtained (or estimated) from the error evolution.

The choice of the constant matrices $K$ and $L$ have impact on the
convergence rate. If $\alpha=\theta$, the error system is $\dot{e} = -
K\,e$, meaning that the error decreases asymptotically to zero
according to a first-order linear dynamics with a time constant
determined by~$K$. If $e=0$, the parameter error has the dynamics
$\dot{\Delta}=-F^T(x)\,L\,F(x)\,\Delta$, and thus the magnitude of $L$
impacts on the convergence rate of the parameters.  Simple choices
for $K$ and $L$ are diagonal matrices with constant values, $K=k\,I$
and $L=l\,I$, for $k$ and $l$ two positive scalars. Thus, the state and
parameter error dynamics become $\dot{e}=-k\,e$ and
$\dot{\Delta}=-l\,F^T(x)\,F(x)\,\Delta$.

Since $F(x)$ is a $(n\times m)$ matrix, its rank is lower or equal
to $\min(n,m)$, and thus the rank of $G(x)$ is also lower and equal
to $\min(n,m)$. However, in order for $G(x)$ to be positive
definite, its rank has to be equal to $m$ (the dimension of the
parameter vector $\theta$), and thus $n\ge m$ is a necessary condition
for $G(x)$ to be full rank. This means that there is an upper bound to
the amount or parameters $m$, in order for convergence to be
guaranteed.  This largely limits the flexibility of the response
system to adapt to arbitrary drive systems, in particular with a large
amount of parameters.

To tackle this problem we propose augmenting the $F(x)$ matrix
with extra rows, as many as needed, in order for $G(x)$ to become full
rank. First, let us designate by $x^*(t)$ a new state vector
consisting in the concatenation of time delayed versions of the
original state vector $x(t)$,
\begin{equation}
  \label{eq:15}
  x^* = \left[\; x_0 \; x_1 \;\cdots\; x_r \;\right]^T,
\end{equation}
where $x_i(t)=x(t-i\,\delta)$, for $i=1\ldots r$ and a $\delta>0$. Using this state
vector, the drive system becomes
\begin{equation}
  \label{eq:16}
  \dot{x}^* = f^*(x^*) + F^*(x^*)\,\theta,
\end{equation}
where
\begin{equation}
  \label{eq:17}
  f^*(x^*)=
  \left[\begin{array}{c}
      f(x_0) \\ \vdots \\ f(x_r)
    \end{array}\right]
  \qquad\mathrm{and}\qquad
  F^*(x^*)=
  \left[\begin{array}{c}
      F(x_0) \\ \vdots \\ F(x_r)
    \end{array}\right].
\end{equation}
This augmented system is equivalent to~\eqref{eq:1}, as the additional
state dimension corresponds to time delayed versions of the original
system.  The response system, with state vector
$y^*\in\reals^{(r+1)n}$ takes the form
\begin{equation}
  \label{eq:18}
  \dot{y}^* = f^*(y^*) + F^*(y^*)\,\alpha + U^*(y^*,x^*,\alpha).
\end{equation}
These two coupled systems~\eqref{eq:16} and~\eqref{eq:18}
with state vectors $x^*$ and $y^*$ can be
viewed as a new pair of drive and response systems by themselves, with error
vector $e^*=y^*-x^*$. Thus, the results obtained above can be directly
applied here: the synchronization controller becomes
\begin{equation}
  \label{eq:5}
  U^*(y^*,x^*,\alpha) = - K^*\,e^* - f^*(y^*) + f^*(x^*) - \left[ F^*(y^*) - F^*(x^*)
  \right]\,\alpha,
\end{equation}
where the matrix $K^*$ can be a $((r+1)n\times (r+1)n)$ block diagonal formed by $K$ matrices,
\begin{equation}
  \label{eq:19}
  K^*=
  \left[\begin{array}{c|c|c}
      K & 0 & 0 \\
      \hline
      0 & K & 0 \\
      \hline
      0 & 0 & \ddots \\
    \end{array}\right].
\end{equation}
The adaptation law becomes then
\begin{equation}
  \label{eq:20}
  \dot{\alpha} = - {F^*}^T(x^*)\,\left[ \left(L^*\,K^*+I\right)\,e^* + L^*\,\dot{e}^* \right],
\end{equation}
where $L^*$ is a $((r+1)n\times (r+1)n)$  matrix, which can also take the form of
a block diagonal in the same fashion as $K^*$ above,
\begin{equation}
  \label{eq:22}
  L^*=
  \left[\begin{array}{c|c|c}
      L & 0 & 0 \\
      \hline
      0 & L & 0 \\
      \hline
      0 & 0 & \ddots \\
    \end{array}\right].  
\end{equation}
If both $K^*$ and $L^*$ have the block diagonal structure as in~\eqref{eq:19} and~\eqref{eq:22}, the adaptation
law~\eqref{eq:20} can be simplified into
\begin{equation}
  \label{eq:21}
 \dot{\alpha} = -\sum_{i=0}^r  F^T(x_i) \,
 [(L\,K+I)\,e_i + \dot{e}_i ],
\end{equation}
where $e_i=y_i-x_i$ and $\dot{e}_i=\dot{y}_i-\dot{x}_i$.

With the above augmented system, we can prove convergence when $n<m$
with the following Corollary:
\begin{corollary}
  \label{cor:1}
  If matrix $G^*(x)=(F^*)(x)^T\,(L^*)^T\,F^*(x)$ is full rank for all
  $x$, then the response system~\eqref{eq:18}, together with the
  adaptation law~\eqref{eq:20}, globally uniformly stabilizes both the
  error system~\eqref{eq:8} at~$e=0$, and the parameter error
  at~$\Delta=0$.
\end{corollary}
\begin{proof}
  The equivalent drive~\eqref{eq:16} and the response~\eqref{eq:18}
  systems satisfy the conditions of Theorem~\ref{thm:2}, as long as
  $G^*(x)$ is full rank.
\end{proof}
\noindent
The rank of $G^*(x)$ cannot be guaranteed \emph{a priori,} but a
necessary condition Corollary can still be stated:
\begin{corollary}
  \label{cor:2}
  If $F$ has rank $n<m$, then $r\ge\lceil\frac{m}{n}-1\rceil$ is
  a necessary condition for $G^*$ to be full rank.
\end{corollary}
\begin{proof}
  The rank of $G^*=(F^*)^T\,(L^*)^T\,F^*$ is at most
  $\min[(r+1)n,m]$. Since $G^*$ is a $m\times m$ matrix, in
  order to be full rank, $(r+1)n\ge m$ has to hold. Therefore,
  $r\ge\frac{m}{n}-1$, but since $r$ is an integer, its lower bound is $\lceil\frac{m}{n}-1\rceil$.
\end{proof}
\noindent
In general, as $r$ is arbitrary, one can expect that there is a value
of $r$ large enough that makes $G^*$ full rank.

Comparing the obtained adaptation law~\eqref{eq:21} with~\eqref{eq:10}
above, one can observe that the gradient of the parameters depends on
several time delayed samples of the error $e$ (as well as on their derivatives $\dot{e}$). A
possible intuition to this result comes from the observation that, if
$m>n$, the degrees of freedom of $e$ are not enough to produce a
meaningful gradient for $\alpha$, if the law~\eqref{eq:10} is employed.
However, with~\eqref{eq:21}, which depends on $e^*$ with $(r+1)n$
degrees of freedom, the gradient of $\alpha$ can have the full
dimensionality of~$m$.

\section{Experimental results}
\label{sec:experimental-results}

This section presents numerical results illustrating the theoretical
results derived above. Two classical chaotic systems were used: the
Lorenz oscillator~\cite{tucker1999lorenz}, commonly used in the
chaotic synchronization literature for numerical
simulations~\cite{kocarev96,pecora97,chen02}, and the Rössler
attractor, designed to behave similarly to the Lorenz system while
being easier to understand~\cite{rossler1976equation}.  Simultaneous
identification and synchronization is simulated, while comparing the
performance of Chen's method~\cite{chen02} with the one proposed here.
For the Chen's method we used controller~\eqref{eq:14} with the
adaptation law~\eqref{eq:12}, and for our method we used
controller~\eqref{eq:5} with the adaptation law~\eqref{eq:21}.

The Lorenz oscillator is a three-dimensional dynamical system that
behaves chaotically for a certain set of
parameters~\cite{tucker1999lorenz}.
In the form
of~(\ref{eq:1}), it can be written as
\begin{equation}
\label{eq:drive}
\left[
  \begin{array}{c}
    \dot{x}\\
    \dot{y}\\
    \dot{z}
  \end{array}
\right] = \left[
  \begin{array}{c}
    0\\
    -y-xz\\
    xy
  \end{array}
\right] + \left[
  \begin{array}{ccc}
    y-x & 0 & 0\\
    0 & x & 0\\
    0 & 0 & -z
  \end{array}
\right] \left[
  \begin{array}{c}
    \theta_1\\
    \theta_2\\
    \theta_3
  \end{array}
\right]
\end{equation}
where $x$, $y$ and $z$ are state variables and $\theta_1$, $\theta_2$
and $\theta_3$ are system parameters.  The Lorenz oscillator was
synchronized with a response system, which is specified by four
parameters. In the form of~(\ref{eq:2}), it can be written as
\begin{equation}
\label{eq:response}
\left[
  \begin{array}{c}
    \dot{u}\\
    \dot{v}\\
    \dot{w}
  \end{array}
\right] = \left[
  \begin{array}{c}
    0\\
    -v-uw\\
    uv
  \end{array}
\right] + \left[
  \begin{array}{cccc}
    v-u & 0 & 0 & 0\\
    0 & u & 0 & 0\\
    0 & 0 & -w & 1
  \end{array}
\right] \left[
  \begin{array}{c}
    \alpha_1\\
    \alpha_2\\
    \alpha_3\\
    \alpha_4
  \end{array}
\right] + \left[
  \begin{array}{c}
    u_1\\
    u_2\\
    u_3
  \end{array}
\right]
\end{equation}
where $u$, $v$ and $w$ are state variables and $\alpha_1$, $\alpha_2$,
$\alpha_3$ and $\alpha_4$ are the parameters.  Note that the rank of
the $F(x)$ matrix in~\eqref{eq:response} is at most~3, while the response
systems uses~4 parameters: $\alpha_4$ is an unnecessary parameter that
is not present in the drive~\eqref{eq:drive}, being artificially
introduced to comparing the two approaches when~$m>n$.  As it was
shown before, under these conditions Chen's method is not guaranteed
to converge, while Corollary~\ref{cor:2} requires~$r\ge 1$ for $G^*$ to be
full rank, and thus a necessary condition for convergence
(as Corollary~\ref{cor:1}).

For this simulation, the classical parameter values for the Lorenz
system were used: $[\theta_1, \theta_2, \theta_3, \theta_4]^T = [10,
28, 8/3, 0]^T$.  The initial states of the drive system and the
controlled system were arbitrarily set to $[8, 9, 10]^T$ and $[3, 4,
5]^T$, respectively.  The parameters of the response system had zero
initial condition. The $L$ and $K$ parameters were set to~$10\,I$
and~$0.1\,I$.

\begin{figure}
\centering
\includegraphics[width=0.8\linewidth]{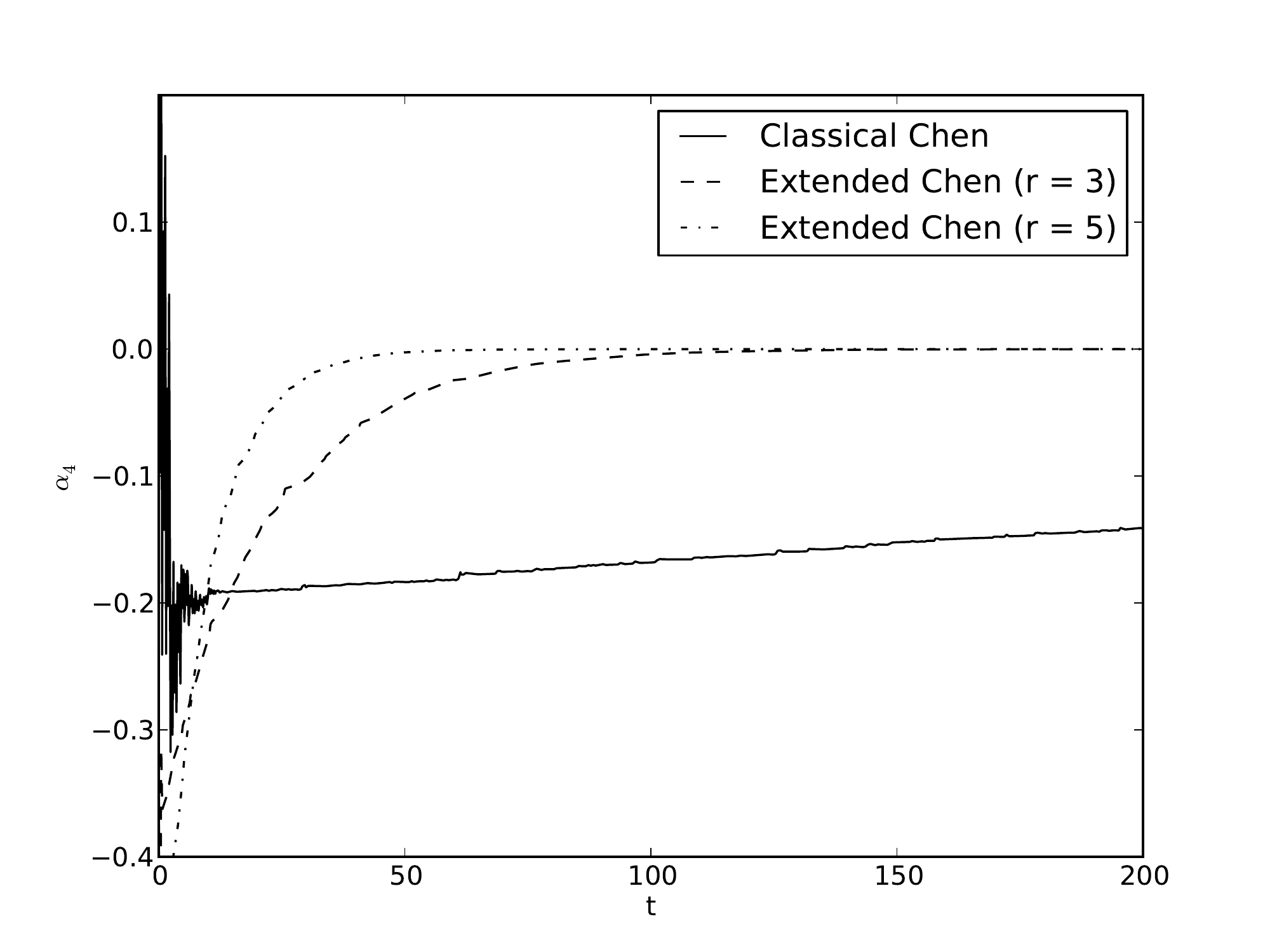}
\caption{Lorenz system: graph of parameter identification results for
  $\alpha_4$. Solid line: Chen's method~\eqref{eq:12}, dotted and
  dash-dotted lines: proposed method~\eqref{eq:21} for $r=3$ and for
  $r=5$.}
\label{fig:lorenzalpha4}
\end{figure}

\begin{figure}
\centering
\includegraphics[width=0.8\linewidth]{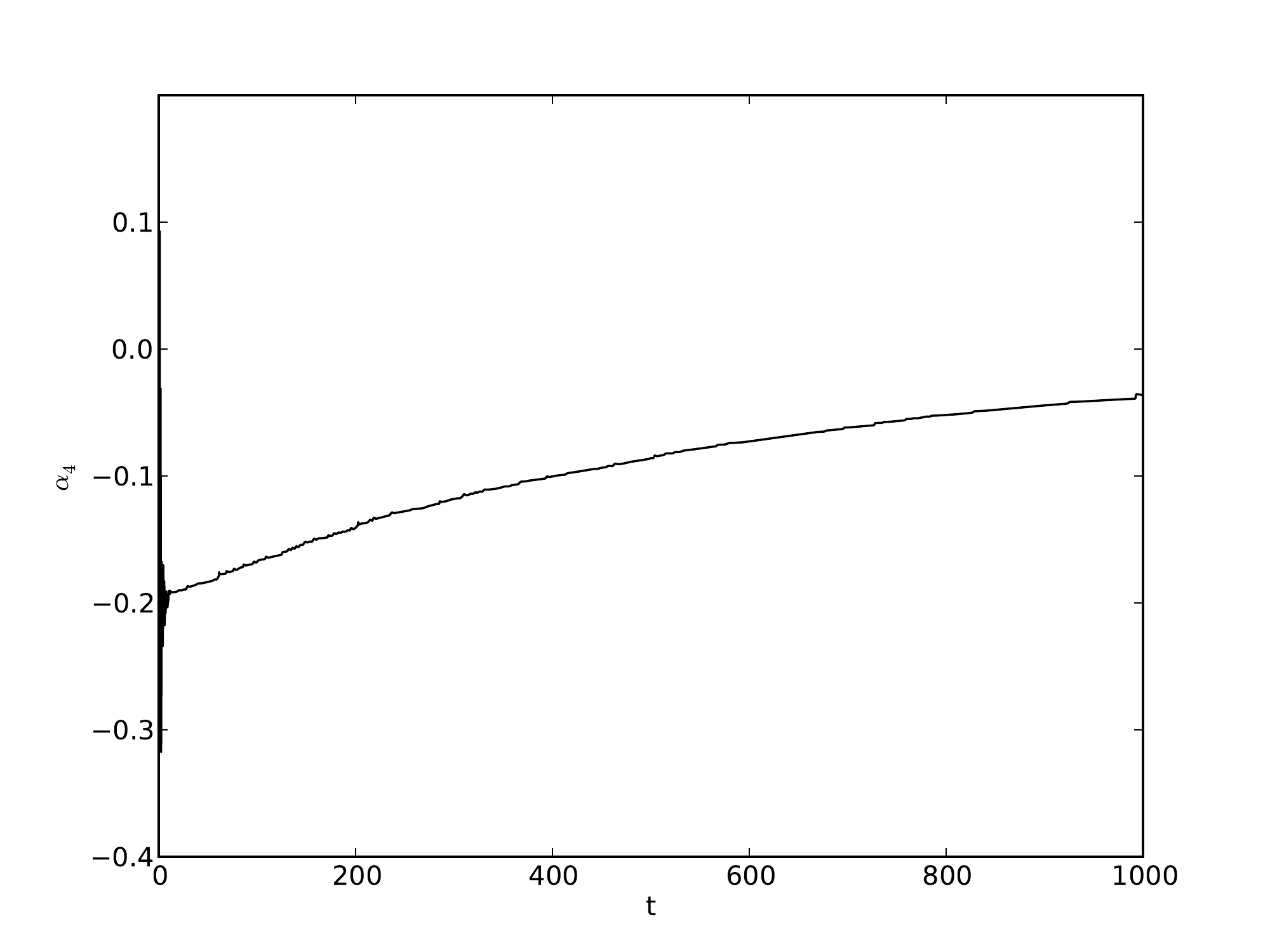}
\caption{Lorenz system: graph of parameter identification results for $\alpha_4$ using Chen's method~\eqref{eq:12}.}
\label{fig:lorenzalpha4chen1000s}
\end{figure}

\begin{figure}
\centering
\includegraphics[width=0.8\linewidth]{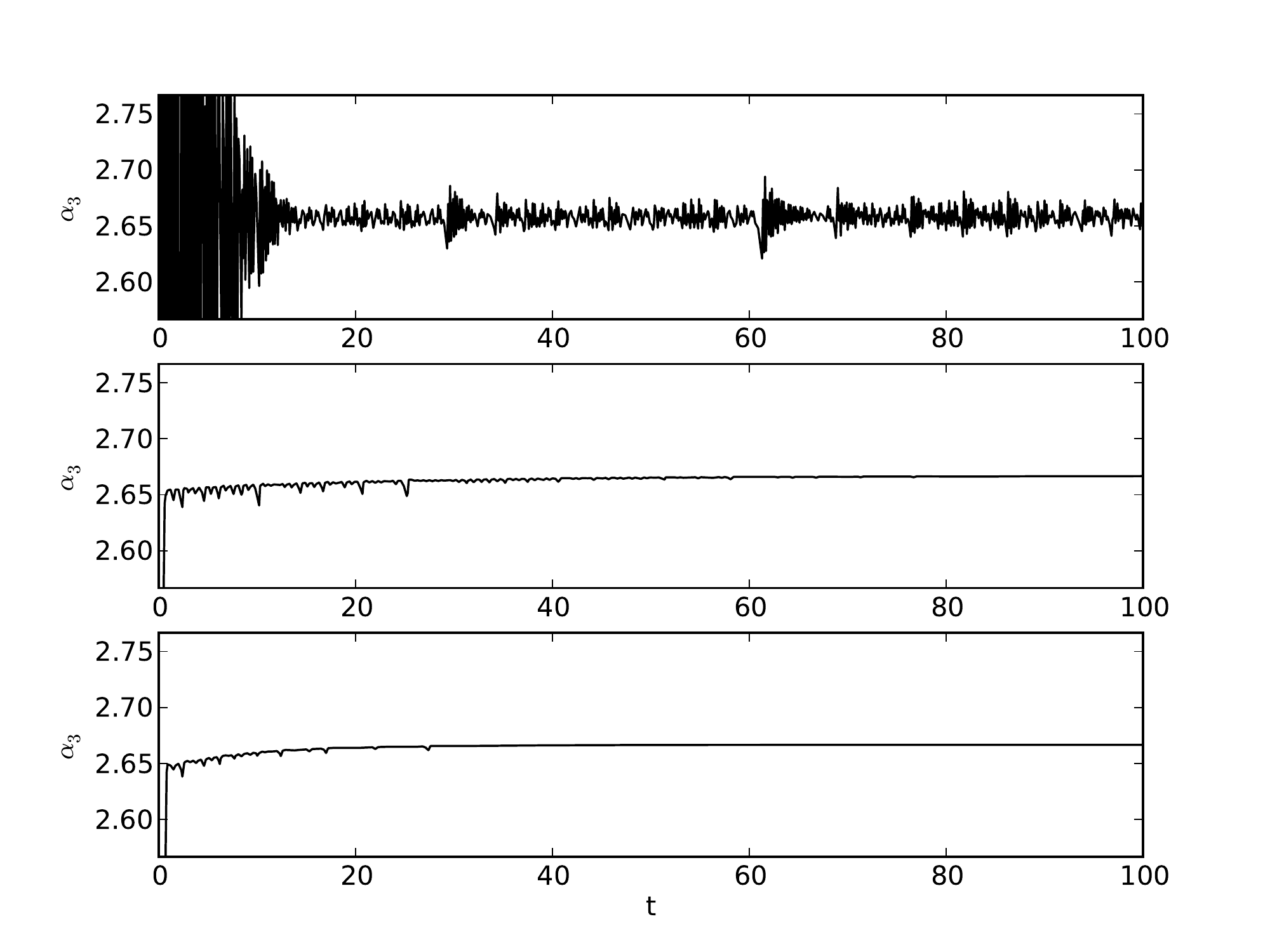}
\caption{Lorenz system: plot of parameter identification results for
  $\alpha_3$. Top plot: Chen's method~\eqref{eq:12}, middle and bottom
  plots: proposed method~\eqref{eq:21} for $r=3$ and for
  $r=5$.}
\label{fig:lorenzalpha3}
\end{figure}

\begin{figure}[h!]
\centering
\includegraphics[width=0.8\linewidth]{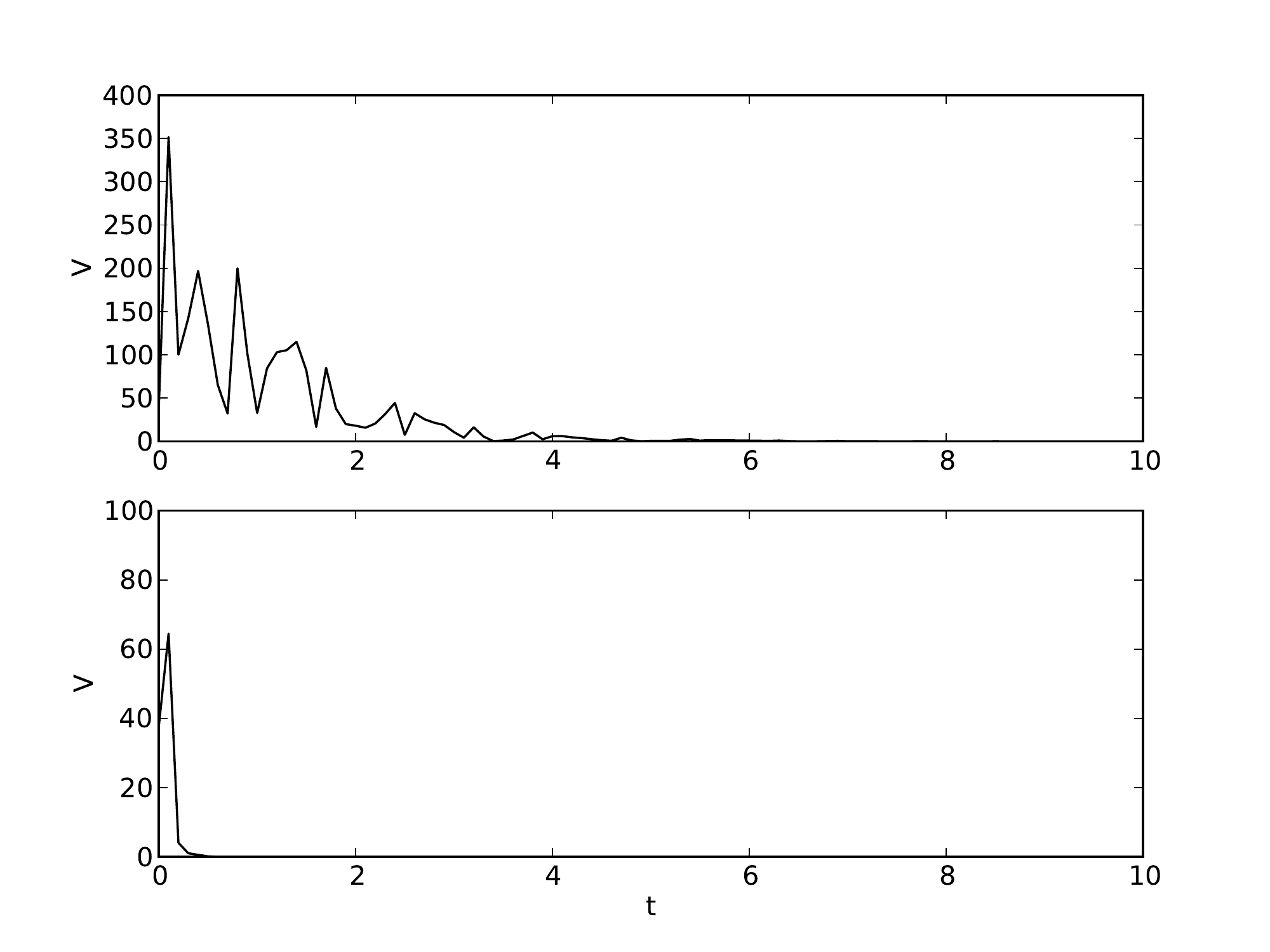}
\caption{Lorenz system: graph of Lyapunov function for synchronization
  error. Top plot: Chen's method~\eqref{eq:12}, bottom plot: proposed
  method~\eqref{eq:21} for $r=3$.}
\label{fig:lorenzv}
\end{figure}

\begin{table}[ht]
\label{tab:lorenzid}
\caption{Time to reach identification error ranges for parameter
  $\alpha_3$ (in simulation seconds)}
\centering
\begin{tabular}{|cc||r|r|r|r|}
\hline
\multicolumn{2}{|c||}{\multirow{2}{*}{Method}} & \multicolumn{4}{c|}{Identification error} \\
\cline{3-6} 
                                          & &   1\% &  0.1\% & 0.01\% & 0.001\% \\
\hline
\multicolumn{2}{|c||}{Classical Chen}     & 203.0 & 1609.8 & 3026.9 & 4640.3  \\
\hline
\multirow{5}{*}{Extended Chen} & r = 1    &   0.4 &  261.5 &  612.1 &  942.7  \\
                               & r = 2    &  10.3 &  104.8 &  226.7 &  308.5  \\
                               & r = 3    &   2.5 &   58.4 &  106.5 &  163.8  \\
                               & r = 4    &   2.5 &   31.1 &   70.9 &   94.8  \\
                               & r = 5    &   2.5 &   21.6 &   43.5 &   67.3  \\
\hline
\end{tabular}
\end{table}

Figure~\ref{fig:lorenzalpha4} shows the numerical results\footnote{All
  simulations were performed using Python together with SciPy and
  PyDDE libraries.} of parameter identification for parameter
$\alpha_4$.  Note the trend for the parameter convergence to be faster
for higher values of $r$.  Figure~\ref{fig:lorenzalpha4chen1000s}
shows the results of parameter identification for the parameter
$\alpha_4$ for Chen's method over a longer time horizon. While Chen's
method is not able to identify this parameter even after $1000$
seconds, our method allows for a significantly faster convergence
(under $200$ seconds).  Figure~\ref{fig:lorenzalpha3} shows the
results of parameter identification for $\alpha_3$.
Table~\ref{tab:lorenzid} shows the time it takes for the parameter
identification error to fall below a percentage of the real parameter
value.  Note again that the convergence is faster for higher values of
$r$.  It is interesting to note that, for instance, during the last
$20$ seconds of the simulation, the coefficient of
variation\footnote{The coefficient of variation is defined as the
  ratio $\sigma/|\mu|$, where $\sigma$ is the standard deviation and
  $\mu$ the sample mean.} of the root mean square error is of
$4.42\times~10^{-3}$ for Chen's method, while for our method it is of
$1.21\times~10^{-4}$ ($r = 3$) and $1.36\times~10^{-6}$ ($r = 5$).
Chen's method is not able to correctly identify this parameter, with
its value oscillating around the true value of $\theta_3$.  Our
method, however, allows for a lower variance in the parameter
identification.  Figure~\ref{fig:lorenzv} shows the synchronization
error, as measured by the Lyapunov function~(\ref{eq:4}).  Both Chen's
method and ours are able to drive the synchronization error to zero.
Our method, however, shows near-instantaneous convergence.  Also, the
magnitude of the error is reduced by comparison to Chen's method.

Similar results were obtained using the Rössler attractor.  In the
form of~(\ref{eq:1}), it can be written as
\begin{equation}
\label{eq:rossler_drive}
\left[
  \begin{array}{c}
    \dot{x}\\
    \dot{y}\\
    \dot{z}
  \end{array}
\right] = \left[
  \begin{array}{c}
    -y-z\\
    x\\
    xz
  \end{array}
\right] + \left[
  \begin{array}{ccc}
    0 & 0 & 0\\
    y & 0 & 0\\
    0 & 1 & -z
  \end{array}
\right] \left[
  \begin{array}{c}
    \theta_1\\
    \theta_2\\
    \theta_3
  \end{array}
\right]
\end{equation}
where $x$, $y$ and $z$ are state variables and $\theta_1$, $\theta_2$ and $\theta_3$ are system parameters.
The Rössler system was synchronized with a response system specified by four parameters. In the form of~(\ref{eq:2}), it can be written as
\begin{equation}
\label{eq:rossler_response}
\left[
  \begin{array}{c}
    \dot{u}\\
    \dot{v}\\
    \dot{w}
  \end{array}
\right] = \left[
  \begin{array}{c}
    -v-w\\
    u\\
    uw
  \end{array}
\right] + \left[
  \begin{array}{cccc}
    0 & 0 & 0 & 0\\
    v & 0 & 0 & 1\\
    0 & 1 & -w & 0
  \end{array}
\right] \left[
  \begin{array}{c}
    \alpha_1\\
    \alpha_2\\
    \alpha_3\\
    \alpha_4
  \end{array}
\right] + \left[
  \begin{array}{c}
    u_1\\
    u_2\\
    u_3
  \end{array}
\right]
\end{equation}
where $u$, $v$ and $w$ are state variables and $\alpha_1$, $\alpha_2$, $\alpha_3$ and $\alpha_4$ are the parameters.
Again, the rank of the $F(x)$ matrix is~3, while the number of parameters is~4.

For this simulation, the commonly used parameter values for the
Rössler system were used: 
$[\theta_1, \theta_2, \theta_3, \theta_4]^T=[0.1, 0.1, 14, 0]^T$.
The initial states of the drive
system and the controlled system were arbitrarily set to $[8, 9,
10]^T$ and $[3, 4, 5]^T$, respectively.  The parameters of the
response system had zero initial condition.  The $L$ and $K$
parameters were set to~$10\,I$ and~$0.1\,I$ (for $I$ being the
identity matrix with appropriate dimensions).

\begin{figure}
\centering
\includegraphics[width=0.8\linewidth]{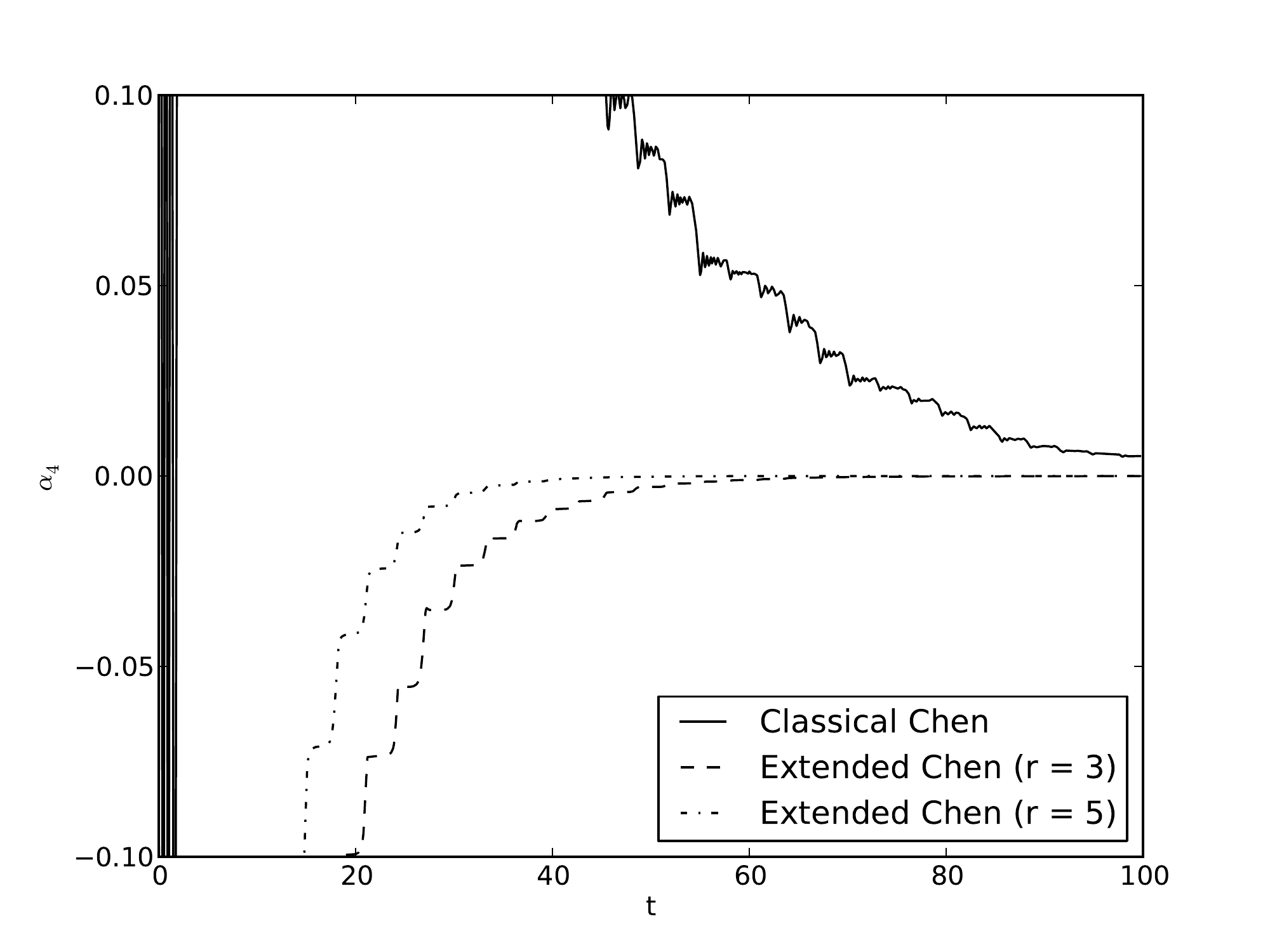}
\caption{Rössler system: graph of parameter identification results for $\alpha_4$. Solid line: Classical Chen, dotted line: Extended Chen ($r = 3$, dash-dotted line: Extended Chen ($r = 5$).}
\label{fig:rossleralpha4}
\end{figure}

\begin{figure}
\centering
\includegraphics[width=0.8\linewidth]{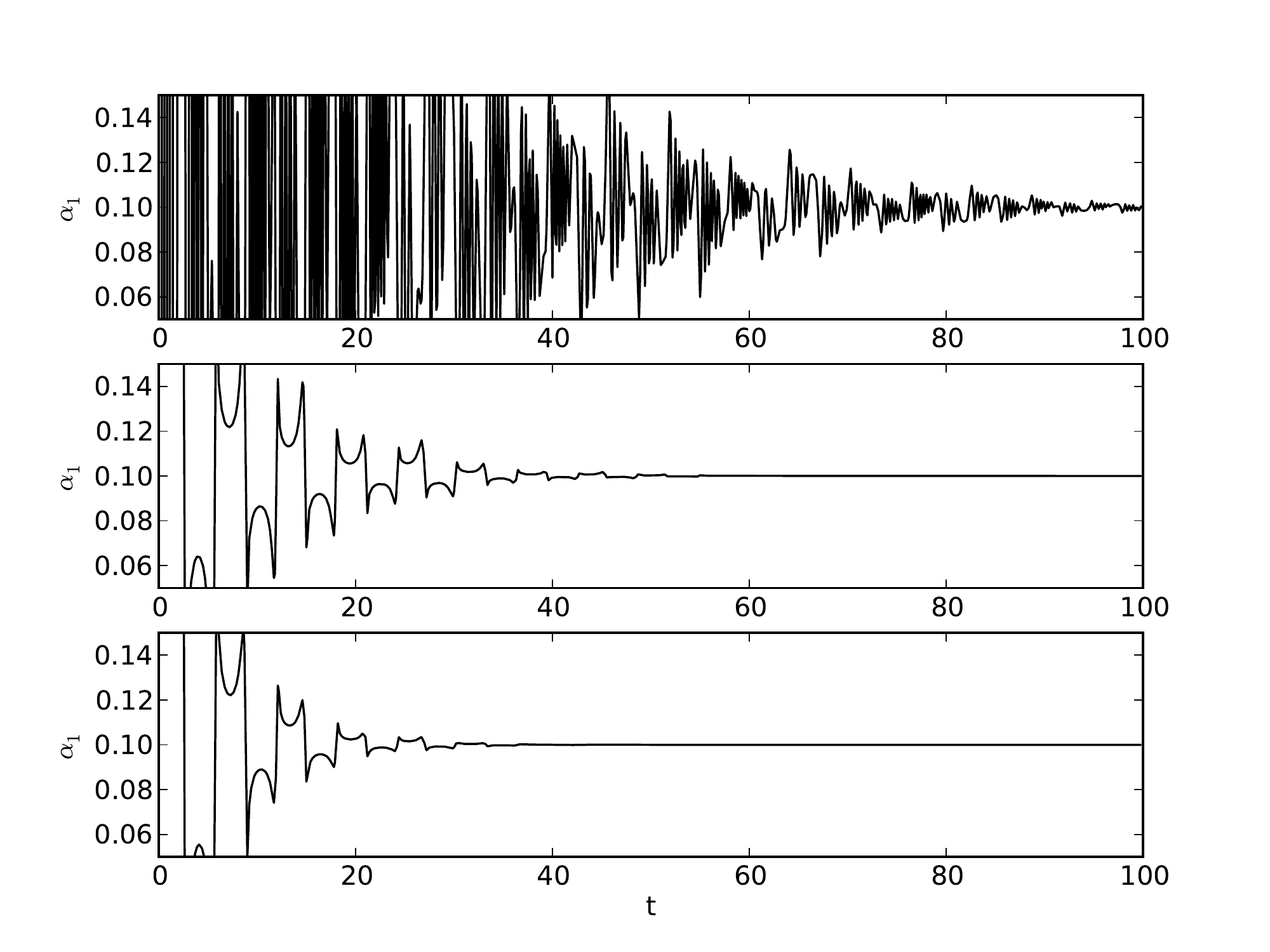}
\caption{Rössler system: plot of parameter identification results for $\alpha_3$. Top plot: Classical Chen, middle plot: Extended Chen ($r = 3$), bottom plot: Extended Chen ($r = 5$).}
\label{fig:rossleralpha1}
\end{figure}

\begin{figure}
\centering
\includegraphics[width=0.8\linewidth]{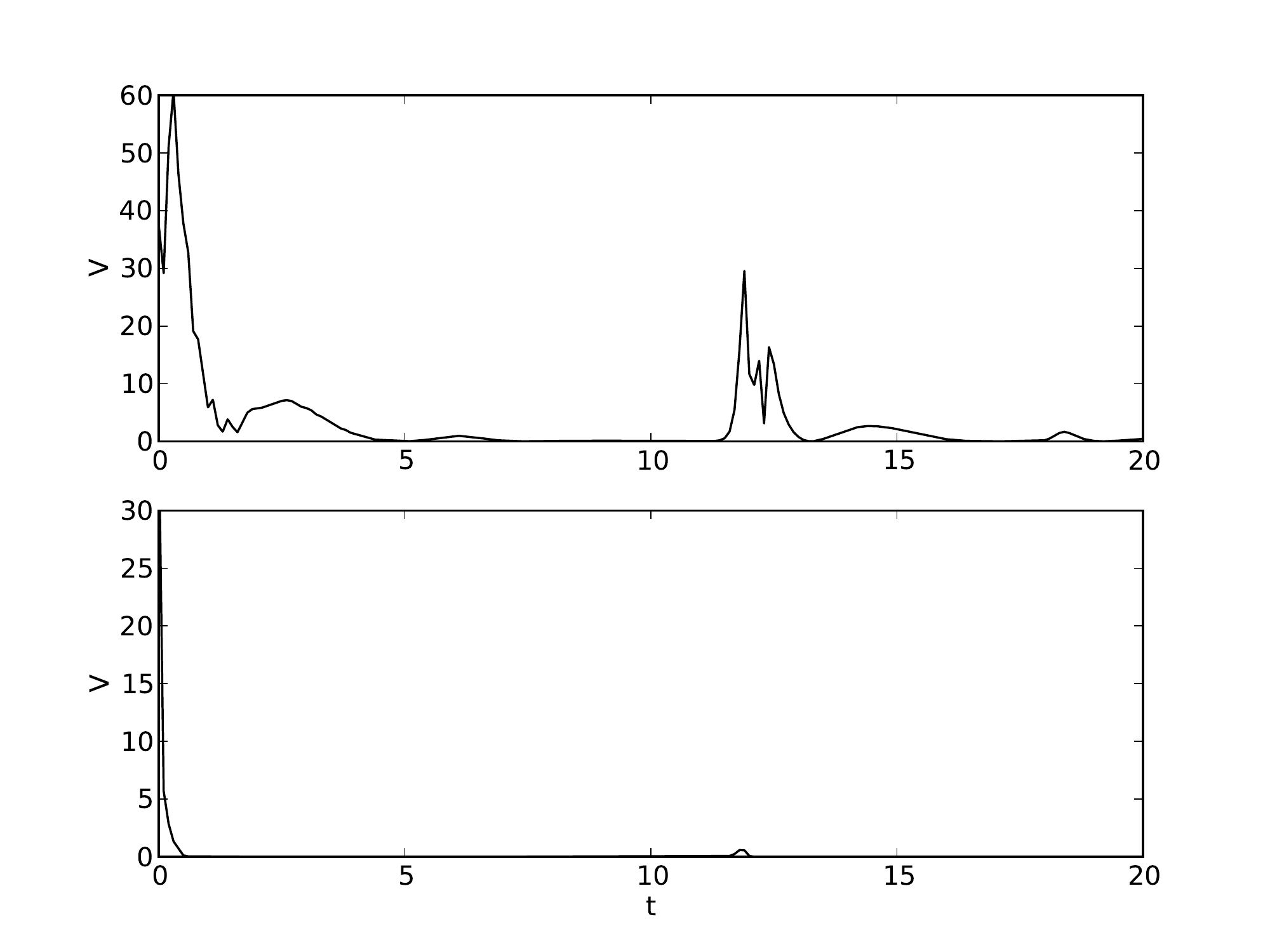}
\caption{Rössler system: graph of Lyapunov function for synchronization error. Top plot: Classical Chen, bottom plot: Extended Chen ($r = 3$).}
\label{fig:rosslerv}
\end{figure}

The improved convergence performance of the proposed method over
Chen's is clearly visible in Figure~\ref{fig:rossleralpha4}, while
parameter convergence is faster for higher values of $r$.
Figure~\ref{fig:rossleralpha1} shows the results of parameter
identification for $\alpha_1$, which, together with $\alpha_4$,
specifies the evolution of the state variable $v$.  During the last
$20$ seconds of the experiment, the coefficient of variation of the
root mean square error is of $2.96\times~10^{-2}$ for Chen's method,
while for our method it is of $6.11\times~10^{-5}$ ($r = 3$) and
$3.32\times~10^{-7}$ ($r = 5$).  Chen's method cannot identify this
parameter correctly, with its value oscillating around the true value
of $\theta_3$.  On the other hand, our method allows for stable
parameter identification.  Again, convergence is faster for greater
values of $r$.  Finally, Figure~\ref{fig:rosslerv} shows the
synchronization error, as measured by the Lyapunov
function~(\ref{eq:4}).  Both methods drive the synchronization error
to zero, while our method shows a significantly faster convergence.
Also, the magnitude of the error is reduced by comparison to Chen's
method.

\section{Conclusions}
\label{sec:conclusions}

Building upon previous work in simultaneous parameters identification
and synchronization of dynamical systems, this paper proposes an
improved method that addresses limitations of the previously published
Chen's method \cite{chen02}. The proposed method is capable of
handling arbitrarily large parameter space dimensions. Convergence
proof of the method is provided, using the Lyapunov's second
method. Numerical results illustrate the proposed method, comparing it
to Chen's and showing better performance in terms of both faster and
less noisy parameter identification.

\bibliography{biblio}
\bibliographystyle{plain}

\end{document}